\documentclass{article}%
\usepackage{amsfonts}
\usepackage{amsmath}
\usepackage{amssymb}
\usepackage{graphicx}
\usepackage[authoryear]{natbib}%
\providecommand{\U}[1]{\protect\rule{.1in}{.1in}}
\providecommand{\U}[1]{\protect\rule{.1in}{.1in}}
\providecommand{\U}[1]{\protect\rule{.1in}{.1in}}
\providecommand{\U}[1]{\protect\rule{.1in}{.1in}}
\newtheorem{theorem}{Theorem}

\newtheorem{corollary}[theorem]{Corollary}

\newtheorem{definition}[theorem]{Definition}
\newtheorem{example}[theorem]{Example}
\newtheorem{examples}[theorem]{Examples}

\newtheorem{lemma}[theorem]{Lemma}

\newtheorem{remark}[theorem]{Remark}

\newenvironment{proof}[1][Proof]{\noindent\textbf{#1.} }{\ \rule{0.5em}{0.5em}}

\newcommand{\Var}{\mbox{\rm Var}}

\begin{document}

\title{
Optimal dual martingales, their analysis  and application to new  algorithms for Bermudan products$^{1,2}$
}

\author{John Schoenmakers \and Junbo Huang \and Jianing Zhang
}

\footnotetext[1]{
 This paper is an extended version of the previous preprint \citet{S1}. Supported by  DFG Research Center \textsc{Matheon}
``Mathematics for Key Technologies'' in Berlin.\\
}

\footnotetext[2]{
Weierstrass Institute for Applied Analysis and Stochastics, Mohrenstr. 39, 10117 Berlin, \{schoenma,huang,zhang\}@wias-berlin.de
}

\maketitle


\begin{abstract}
In this paper we introduce and study the concept of optimal and surely optimal
dual martingales in the context of dual valuation of Bermudan options, and outline
the development of new algorithms in this context. We provide a characterization theorem, a theorem which gives conditions for a martingale to be surely
optimal, and a stability theorem concerning
martingales which are near to be surely optimal in a sense. Guided by these results we
develop a framework of backward algorithms for constructing such a martingale.
  In turn this martingale may then be utilized for computing
an upper bound of the Bermudan product. The methodology  is purely dual in the sense that it doesn't require certain input approximations to the Snell envelope.

In an It\^o-L\'evy environment we outline a  particular regression based backward algorithm which allows for computing dual upper bounds without nested
Monte Carlo simulation. Moreover, as a by-product this algorithm also
provides approximations to the continuation values of the product, which in turn determine a stopping policy. Hence,  we may obtain lower bounds at the same time.

In a first numerical study we demonstrate the backward dual regression algorithm in a Wiener environment at well known benchmark examples.  It turns out
that the method  is at least comparable to the one in \citet{BBS}  regarding accuracy,  but regarding computational robustness there are even
several  advantages.
\\[0.2cm]\emph{Keywords:} Bermudan options, duality,
Monte Carlo simulation, linear regression, surely optimal martingales, backward algorithm. \\[0.2cm]\emph{MSC:} 62L15, 65C05, 91B28
\end{abstract}

\section{Introduction}

It is well-known that the evaluation of Bermudan callable derivatives comes down
to solving an optimal stopping problem. For many callable exotic products, e.g. interest products, the underlying state space is high-dimensional
however. As such these products are usually computationally expensive to solve with
deterministic (PDE) methods and therefore simulation based (Monte Carlo)
methods are called for. The first developments in this respect concentrated on
the construction of a ``good'' exercise policy. We mention, among others,
regression based methods by \citet{Car}, \citet{LS}, and \citet{TV}, the
stochastic mesh method of \citet{BroGla04}, and quantization algorithms by
\citet{BalPag03}. Especially for very high dimensions, \citet{KS} developed a
policy improvement approach which can be effectively combined with \citet{LS}
for example (see \citet{BKS} and \citet{BKS1}).

As a common feature, the aforementioned simulation methods provide lower
biased estimates for the Bermudan product under consideration. As a new
breakthrough, \citet{Ro}, and \citet{HK} introduced a dual approach, which
comes down to minimizing over a set of martingales rather than maximizing
over a family of stopping times. By its very nature the dual approach gives
upper biased estimates for the Bermudan product and after its discovery
several numerical algorithms for computing dual upper bounds have been proposed.
Probably the most popular one is the method of \citet{AB}, although this
method requires nested Monte Carlo simulation (see also \citet{KS} and
\citet{S}). In a Wiener environment, \citet{BBS} provides a fast generic method
for computing dual upper bounds which avoids nested simulations. Further \citet{BrSmSu} consider
dual optimization via enlarging the information were an exercise decision may depend on. In this setting they also provide
an example were a tight dual upper bound can be obtained by non-nested simulation.

The algorithms for computing dual upper bounds so far have in common that they
start with some given ``good enough'' approximation of the Snell envelope and
then construct the Doob martingale due to this approximation. In a recent
paper \citet{Ro1}, points out how to construct a particular 'good' martingale
via a sequence of martingales which are constant on an even bigger time
interval. In this construction no input approximation to the Snell envelope is
used.
The methods proposed in this paper have some flavor of the method of \citet{Ro1}, in the sense that no approximation to the Snell envelope is
involved either. In a recent paper \citet{Desai} treat the dual problem by methods from convex optimization theory.

The structure of this paper is as follows. Starting
with a short resume of well-known facts
on Bermudan derivatives in Section~\ref{secberm}, we analyze in
Section~\ref{sect_sure} the almost sure property of
the dual representation in detail. There we introduce the concept of a surely optimal martingale, which is loosely
speaking, a martingale that minimizes the dual representation with a
particular almost sure property. In this respect we
will point out that a martingale which minimizes the dual representation is
not necessarily surely optimal, and on the other hand, a surely optimal
martingale is generally not unique.

In Section~\ref{charsec} we present, as one of the main contributions of this
paper, a characterization theorem for surely optimal martingales (Theorem~\ref{surechar}). Moreover, we provide another result that guarantees that a martingale
is surely optimal if it satisfies a certain measurability criterion (Theorem~\ref{sureth}).

In  applications of the algorithm of \citet{AB} one generally observes that the lower the variance of the upper bound estimator, i.e. the closer the
corresponding martingale is to a surely optimal one,
the sharper is the corresponding dual upper bound.
Actually this
observation was not well studied from a mathematical point of view so far. In
Section~\ref{secnear} we study this phenomenon and, as a next main contribution,
give an explanation of it by Theorem~\ref{near} and Corollary~\ref{stab1}. In fact, the latter
corollary may be considered a stability statement connected to Theorem~\ref{sureth}.

Guided by the new theoretical insights we develop in
Section~\ref{secalg}
algorithms  for constructing dual martingales that are based on minimization
of the  variance (respectively expected conditional variance) of  corresponding dual representations and estimators.
In this context we present in an It\^o-L\'evy environment a regression based backward procedure that constructs a dual martingale via  minimizing backwardly  in time the expected (conditional) variances of the  dual estimators corresponding to the Snell envelope.
We so obtain  a martingale that allows for computing upper bounds without nested Monte Carlo (like in \citet{BBS}). Moreover we obtain, as a
by-product,  estimations of continuation
values. Thus, as a result, we end up with a procedure that computes upper
bounds as well as lower bounds simultaneously via a non-nested simulation procedure. The
procedure is quite easy to implement and may be considered as a valuable  alternative to the non-nested method of \citet{BBS}, where a dual martingale
is obtained by constructing  a discretized  Clark-Ocone derivative of some (input) approximation to the Snell envelope via regression. In particular, our new
procedure only requires regression at each exercise  date, in contrast to the procedure of  \citet{BBS} that requires regression at each time point
of a sufficient refinement of the exercise grid.

In Section~\ref{newsecnum}, we present a numerical study of our algorithm. We illustrate at two
multi-dimensional benchmark products (one of which is also considered in \citet{BBS})  a backward regression algorithm  that,
regarding accuracy and computational effort, produces upper bounds that show to be at least
of the same quality as those in \citet{BBS},  and    fast lower bounds that are overall better than in \cite{BBS} moreover.
 In an Appendix, we provide standard results from Statistics which are used for several technical arguments in Section \ref{secalg0}.

\section{Bermudan derivatives and optimal stopping}

\label{secberm}

Let $(Z_{i}:$ $i=0,1,\ldots,T)$\footnote{For notational convenience we have
chosen for this stylized time set. The reader may reformulate all statements
and results in this paper for a general discrete time set $\{T_{0}%
,T_{1},\ldots, T_{J}\}$ in a trivial way.} be a non-negative stochastic
process in discrete time on a filtered probability space $(\Omega
,\mathcal{F},P),$ adapted to a filtration $\mathbb{F}:=(\mathcal{F}%
_{i}:\,0\leq i\leq T)$ which satisfies $E|Z_{i}|<\infty,$ for $0\leq i\leq T.$
The measure $P$ may be considered as a pricing measure and the process $Z$ may
be seen as a (discounted) cash-flow which an investor may exercise once in
the time set $\{0,...,T\}.$ Hence, she is faced with a Bermudan product. A well-known fact is that a fair price of such a derivative is given by the Snell
envelope
\begin{equation}
Y_{i}^{\ast}=\sup_{\tau\in\{i,...,T\},}E_iZ_{\tau},\text{ \ \ }0\leq i\leq T,
\label{stop}%
\end{equation}
at time $i=0.$ In (\ref{stop}), $\tau$ denotes a stopping time, $E_{i}%
:=E_{\mathcal{F}_{i}}$ denotes the conditional expectation with respect to the
$\sigma$-algebra $\mathcal{F}_{i},$ and $\sup$ ($\inf$) is to be understood as
\textit{essential supremum} (\textit{essential infimum}) if it ranges over an
uncountable family of random variables. Let us recall some well-known facts
(e.g. see \citet{Ne1}).

\begin{enumerate}
\item The Snell envelope $Y^{*}$ of $Z$ is the smallest super-martingale that
dominates $Z$.

\item A family of optimal stopping times is given by
\[
\tau_{i}^{\ast}=\inf\{j:j\geq i,\quad Z_{j}\geq Y_{j}^{\ast}\},\quad0\leq
i\leq T.
\]
In particular,
\[
Y_{i}^{\ast}=E_iZ_{\tau_{i}^{\ast}},\text{ \ \ }0\leq i\leq T,
\]
and the above family is the family of first optimal stopping times if several
optimal stopping families exist.
\end{enumerate}

The optimal stopping problem (\ref{stop}) has a natural interpretation from the
point of view of the option holder: she seeks for an optimal exercise strategy
which optimizes her expected payoff. On the other hand, the seller of the
option rather seeks for the minimal cash amount (smallest supermartingale) he
has to have at hand in any case the holder of the option exercises.

\section{Duality and surely optimal martingales}

\label{sect_sure}

We briefly recall the dual approach proposed by \citet{Ro} and, independently,
\citet{HK}. The dual approach is based on the following observation: for any
martingale $(M_{j})$ with $M_{0}=0$ we have
\begin{equation}
Y_{0}^{\ast}=\sup_{\tau\in\{0,\ldots,T\}}E_{0}Z_{\tau}\leq\sup_{\tau
\in\{0,\ldots,T\}}E_{0}\left(  Z_{\tau}-M_{\tau}\right)  \leq E_{0}\max_{0\leq
j\leq T}\left(  Z_{j}-M_{j}\right),  \label{Rogers}%
\end{equation}
hence the right-hand side provides an upper bound for $Y_{0}^{\ast}$.
\citet{Ro} and \citet{HK} showed that (\ref{Rogers})  holds with equality
 for the martingale part of the Doob decomposition of
$Y^{\ast},$ i.e. $Y_{j}^{\ast}=Y_{0}^{\ast}+M_{j}^{\ast}-A_{j}^{\ast},$ where
$M^{\ast}$ is a martingale with $M_{0}^{\ast}=0,$ and $A^{\ast}$ is
predictable with $A_{0}^{\ast}=0.$ More precisely we have
\begin{equation}
M_{j}^{\ast}=\sum\limits_{l=1}^{j}\left(  Y_{l}^{\ast}-E_{l-1}Y_{l}^{\ast
}\right)  ,\quad A_{j}^{\ast}=\sum\limits_{l=1}^{j}\left(  Y_{l-1}^{\ast
}-E_{l-1}Y_{l}^{\ast}\right)  , \label{AddMart}%
\end{equation}
from which we see $A^{\ast}$ is non-decreasing due to $Y^\ast$ being a supermartingale. In addition,
they showed that%
\begin{equation}
Y_{0}^{\ast}=\max_{0\leq j\leq T}\left(  Z_{j}-M_{j}^{\ast}\right)  \text{
\ \ a.s.} \label{alsu}%
\end{equation}
The next lemma, by \citet{KS2}, provides a somewhat more general class of
supermartingales, which turns relation (\ref{Rogers}) into an equality such that
moreover (\ref{alsu}) holds.

\begin{lemma}
\label{AddDual} Let $S$ be a supermartingale  with
$S_{0}=0,$. Assume that $Z_{j}-Y_{0}^{\ast}\leq S_{j}$,
$1\leq j\leq T$. It then holds that
\begin{equation}
Y_{0}^{\ast}=\max_{0\leq j\leq T}(Z_{j}-S_{j})\quad a.s. \label{AddDualEq}%
\end{equation}
For the proof see{ }\citet{KS2}.
\end{lemma}

\begin{examples}
\label{ex} Obviously, by taking for $S$ the Doob martingale as constructed in (\ref{AddMart}),
Lemma \ref{AddDual} applies. However, the Doob martingale is not the only one. For example, in
the case $Z>0$ a.s. we may also take
\[
S_{j}=(N_{j}^{\ast}-1)Y_{0}^{^{\ast}},%
\]
where $N^{\ast}$ is the multiplicative Doob part of the Snell envelope. More precisely,
$Y_{j}^{\ast}=Y_{0}^{\ast}N_{j}^{\ast}B_{j}^{\ast}$ for a martingale $N^{\ast
}$ with $N_{0}^{\ast}=1$ and predictable $B^{\ast}$ with $B_{0}^{\ast}=1.$
Hence
\begin{equation}
\label{muld}N_{j}^{\ast}=%
{\displaystyle\prod\limits_{l=1}^{j}}
\frac{Y_{l}^{\ast}}{E_{l-1}Y_{l}^{\ast}},\quad B_{j}^{\ast}=%
{\displaystyle\prod\limits_{l=1}^{j}}
\frac{E_{l-1}Y_{l}^{\ast}}{Y_{l-1}^{\ast}}.
\end{equation}
Indeed, since $B^{\ast}$ is non-increasing due to $Y^\ast$ being a supermartingale, we
have%
\[
S_{j}=Y_{0}^{^{\ast}}\left(  \frac{Y_{j}^{\ast}}{Y_{0}^{\ast}B_{j}^{\ast}%
}-1\right)  \geq Y_{0}^{^{\ast}}\left(  \frac{Y_{j}^{\ast}}{Y_{0}^{\ast}%
}-1\right)  =Y_{j}^{\ast}-Y_{0}^{^{\ast}}\geq Z_{j}-Y_{0}^{^{\ast}},%
\]
thus, Lemma \ref{AddDual} applies again.
\end{examples}

The multiplicative Doob decomposition in (\ref{muld}) is used by \citet{Jam}
for constructing a multiplicative dual representation. In a comparative study,
\citet{GC} pointed out however, that from a numerical point of view additive
dual algorithms perform better due to the nice almost sure property
(\ref{alsu}).

\begin{remark}
\label{counter} It is \textbf{not} true that for any martingale $M$ which
turns (\ref{Rogers}) into equality the almost sure statement (\ref{alsu})
holds. As a simple counterexample, consider $T=1,$ $Z_{0}=0,$ $Z_{1}=2,$
$M_{0}=0,$ and $M_{1}=\pm1$ each with probability $1/2.$ Indeed, we see that
$Y_{0}^{\ast}=2=E_{0}(2-M_{1})=E_{0}\max(0,2-M_{1}),$ but, $Y_{0}^{\ast}%
\neq\max(0,2-M_{1})$ a.s.
\end{remark}

In order to have a unified dual representation for the Snell envelope
$Y_{i}^{\ast}$ at any $i,$ it is convenient to drop the assumption that
martingales start at zero. We then may restate the dual theorem as
\begin{align}
Y_{i}^{\ast}  &  =\inf_{M\in\mathcal{M}}E_{i}\max_{i\leq j\leq T}\left(
Z_{j}-M_{j}+M_{i}\right) \label{uni}\\
&  =\max_{i\leq j\leq T}\left(  Z_{j}-M_{j}^{\ast}+M_{i}^{\ast}\right)  \text{
\ a.s.,} \label{uni1b}%
\end{align}
for any $i,$ $0\leq i\leq T,$ where $\mathcal{M}$ is the set of all
martingales and $M^{\ast}$ is the Doob martingale part of $Y^{\ast}.$

In view of Remark~\ref{counter} and Examples~\ref{ex}, a martingale for which
the infimum (\ref{uni}) is attained must not necessarily satisfy an almost
sure property such as (\ref{uni1b}), and, martingales which do satisfy such almost sure property are
generally not unique. We hence propose the following concept of \emph{surely optimal} martingales.

\begin{definition}
We say that a martingale $M$ is \textbf{surely optimal} for the Snell envelope
$Y^{\ast}$ at a time $i,$ $0\leq i\leq T,$ if it holds
\begin{equation}
Y_{i}^{\ast}=\max_{i\leq j\leq T}\left(  Z_{j}-M_{j}+M_{i}\right)  \text{
\ a.s.}  \label{uni1a}
\end{equation}
\end{definition}

\begin{remark}
Obviously, the Doob martingale of $Y^{*}$ is {surely optimal} at
each $i,$ $0\le i\le T,$ and any martingale $M$ is trivially surely optimal at $i=T.$
However, it is
\textbf{not} true that sure optimality for some $i$ with $i<T$ implies sure
optimality at $i+1.$ As a counterexample let us consider $T=2,$ and $Z_{0}=4,$
$Z_{1}=0,$ $Z_{2}=2.$ Take as martingale $M_{0}=0,$ $M_{1}=\pm1,$ each with
probability $1/2,$ and $M_{2}=M_{1}\pm1,$ each with probability $1/2$
conditional on $M_{1}.$ Then $\max_{0\leq j\leq2}\left(  Z_{j}-M_{j}%
+M_{0}\right)  =4$ a.s. Since we have trivially $Y_{0}^{\ast
}=4,$ $M$ is surely optimal at $i=0.$ But, $\max_{1\leq j\leq2}\left(
Z_{j}-M_{j}+M_{1}\right)  =2-M_{2}+M_{1}\notin\mathcal{F}_{1},$ so $M$ is not
surely optimal for $Y^{\ast}$ at $i=1.$
\end{remark}

\section{Characterization of surely optimal martingales}\label{charsec}
In this section we give a characterization of martingales that are surely optimal for
all $i=0,\ldots,T.$
\begin{theorem}
\label{surechar} A martingale $M$ with $M_{0}=0$ is surely optimal for
$i=0,\ldots,T,$ if and only if there exists a sequence of adapted random
variables $(\zeta_{i})_{0\leq i\leq T,}$ such that $E_{i-1}\zeta_{i}=1,$ and
$\zeta_{i}\geq0$ for all $0<i\leq T,$ and
\begin{equation}
M_{i}=M_{i}^{\ast}-A_{i}^{\ast}+\sum_{l=1}^{i}\left(  A_{l}^{\ast}%
-A_{l-1}^{\ast}\right)  \zeta_{i}, \label{unrep}%
\end{equation}
where, respectively, $M^{\ast}$ is the Doob martingale and $A_{i}^{\ast}$ the
predictable process of the Snell envelope $Y^{\ast}$ as given in (\ref{AddMart}).
\end{theorem}

\begin{proof}
i) Let us assume that $M$ is surely optimal as stated. Then by
(\ref{uni1a}) it holds for any $0<i\leq T,$%
\begin{align}
Y_{i-1}^{\ast}  &  =\max_{i-1\leq j\leq T}\left(  Z_{j}-M_{j}+M_{i-1}\right)
\nonumber\\
&  =\max(Z_{i-1},M_{i-1}-M_{i}+\max_{i\leq j\leq T}\left(  Z_{j}-M_{j}%
+M_{i}\right)  )\nonumber\\
&  =\max(Z_{i-1},M_{i-1}-M_{i}+Y_{i}^{\ast}). \label{df}%
\end{align}
Since $Z_{i-1}\leq Y_{i-1}^{\ast},$ and since $Z_{i-1}<Y_{i-1}^{\ast}$
implies $A_{i-1}^{\ast}=A_{i}^{\ast},$ we obtain from (\ref{df}) and the Doob decomposition $Y^\ast_i = Y^\ast_0 + M^\ast_i - A^\ast_i$
\begin{align*}
Y_{i-1}^{\ast}-Z_{i-1}  &  =\left(  M_{i-1}-M_{i}+Y_{i}^{\ast}-Z_{i-1}\right)
^{+}\\
&  =\left(  M_{i-1}-M_{i}+M_{i}^{\ast}-M_{i-1}^{\ast}-A_{i}^{\ast}%
+A_{i-1}^{\ast}+Y_{i-1}^{\ast}-Z_{i-1}\right)  ^{+}\\
&  =1_{Z_{i-1}<Y_{i-1}^{\ast}}\left(  M_{i-1}-M_{i}+M_{i}^{\ast}-M_{i-1}%
^{\ast}+Y_{i-1}^{\ast}-Z_{i-1}\right)  ^{+}\\
&  +1_{Z_{i-1}=Y_{i-1}^{\ast}}\left(  M_{i-1}-M_{i}+M_{i}^{\ast}-M_{i-1}%
^{\ast}-A_{i}^{\ast}+A_{i-1}^{\ast}\right)  ^{+}.
\end{align*}
So we must have%
\begin{gather*}
1_{Z_{i-1}<Y_{i-1}^{\ast}}\left(  Y_{i-1}^{\ast}-Z_{i-1}\right)  =\\
1_{Z_{i-1}<Y_{i-1}^{\ast}}\left(  M_{i-1}-M_{i}+M_{i}^{\ast}-M_{i-1}^{\ast
}+Y_{i-1}^{\ast}-Z_{i-1}\right)  ,\text{ \ \ and}\\
1_{Z_{i-1}=Y_{i-1}^{\ast}}\left(  M_{i-1}-M_{i}+M_{i}^{\ast}-M_{i-1}^{\ast
}-A_{i}^{\ast}+A_{i-1}^{\ast}\right)  ^{+}=0,
\end{gather*}
respectively. Hence we get%
\begin{gather}
1_{Z_{i-1}<Y_{i-1}^{\ast}}\left(  M_{i-1}-M_{i}+M_{i}^{\ast}-M_{i-1}^{\ast
}\right)  =0,\text{ \ \ and}\label{sc}\\
1_{Z_{i-1}=Y_{i-1}^{\ast}}\left(  M_{i-1}-M_{i}+M_{i}^{\ast}-M_{i-1}^{\ast
}-A_{i}^{\ast}+A_{i-1}^{\ast}\right)  =-1_{Z_{i-1}=Y_{i-1}^{\ast}}\mu_{i},
\label{sc0}%
\end{gather}
for some non-negative $\mathcal{F}_{i}$-measurable random variable $\mu_{i}.$
W.l.o.g. we assume that $\mu_{i}\equiv0$ on the set $\{Z_{i-1}<Y_{i-1}^{\ast
}\}.$ By taking $\mathcal{F}_{i-1}$ conditional expectations on both sides of
(\ref{sc0}), and using the martingale property of both $M$ and $M^{\ast},$ and
the predictability of $A^{\ast}$, it then follows that
\begin{equation}
E_{i-1}\mu_{i}=1_{Z_{i-1}=Y_{i-1}^{\ast}}E_{i-1}\mu_{i}=1_{Z_{i-1}%
=Y_{i-1}^{\ast}}\left(  A_{i}^{\ast}-A_{i-1}^{\ast}\right)  . \label{part}%
\end{equation}
In particular, since $\mu_{i}\geq0$ almost surely, it follows from
(\ref{part}) that $\mu_{i}=0$ on the set $\{A_{i}^{\ast}=A_{i-1}^{\ast}\}$
(in which $\{Z_{i-1}<Y_{i-1}^{\ast}\}$ is contained as a subset). We next define
\begin{equation}
\zeta_{i}:=\begin{cases}
\left(  A_{i}^{\ast}-A_{i-1}^{\ast}\right)  ^{-1}\mu_{i}, &\text{if }A_{i}^{\ast}>A_{i-1}^{\ast},\\
1, &\text{else,}%
\end{cases}\label{ze}%
\end{equation}
and we see that we have a.s. $\mu_{i}=\left(  A_{i}^{\ast}-A_{i-1}^{\ast}\right) \zeta_{i}$. By (\ref{part}) we have (using the convention $0\cdot\infty=0$)
\begin{align*}
E_{i-1}\zeta_{i}  &  =1_{A_{i}^{\ast}>A_{i-1}^{\ast}}E_{i-1}\,\,\left(
A_{i}^{\ast}-A_{i-1}^{\ast}\right)  ^{-1}\mu_{i}+1_{A_{i}^{\ast}=A_{i-1}%
^{\ast}}\\
&  =1_{A_{i}^{\ast}>A_{i-1}^{\ast}}1_{Z_{i-1}=Y_{i-1}^{\ast}}+1_{A_{i}^{\ast
}>A_{i-1}^{\ast}}1_{Z_{i-1}<Y_{i-1}^{\ast}}+1_{A_{i}^{\ast}=A_{i-1}^{\ast}}=1,
\end{align*}
since the middle term is trivially zero. We thus obtain from (\ref{sc}) and
(\ref{sc0})
\[
M_{i-1}-M_{i}+M_{i}^{\ast}-M_{i-1}^{\ast}-A_{i}^{\ast}+A_{i-1}^{\ast}=-\left(
A_{i}^{\ast}-A_{i-1}^{\ast}\right)  \zeta_{i},
\]
from which (\ref{unrep}) follows.

\medskip

ii) Conversely, if a martingale
$M$ satisfies (\ref{unrep}), we have for any $0\leq i\leq T,$%
\begin{align*}
\max_{i\leq j\leq T}\left(  Z_{j}-M_{j}+M_{i}\right)   &  =\max_{i\leq j\leq
T}\left(  Z_{j}-M_{j}^{\ast}+A_{j}^{\ast}-\sum_{l=1}^{j}\left(  A_{l}^{\ast
}-A_{l-1}^{\ast}\right)  \zeta_{l}\right. \\
&  \left.  +M_{i}^{\ast}-A_{i}^{\ast}+\sum_{l=1}^{i}\left(  A_{l}^{\ast
}-A_{l-1}^{\ast}\right)  \zeta_{l}\right) \\
&  =Y_{i}^{\ast}+\max_{i\leq j\leq T}\left(  Z_{j}-Y_{j}^{\ast}-\sum
_{l=i+1}^{j}\left(  A_{l}^{\ast}-A_{l-1}^{\ast}\right)  \zeta_{l}\right)  \leq
Y_{i}^{\ast},
\end{align*}
and then by (\ref{uni}) the almost sure optimality follows.
\end{proof}
\\ \ \\
By Theorem~\ref{surechar} we have immediately the following alternative
characterization of almost sure martingales. It basically says that a martingale is surely optimal if the Snell envelope can be representated in a way that resembles the Doob decomposition but where the predictable process is replaced by a process which is in general only adapted.

\begin{corollary}
A martingale $M$ with $M_{0}=0$ is surely optimal for $i=0,\ldots,T,$ if and
only if there exists an non-decreasing adapted process $N$ with $N_{0}=0$ such that\footnote{Note that $N$ is not assumed to be predictable.}%
\[
Y_{i}^{\ast}=Y_{0}^{\ast}+M_{i}-N_{i}.
\]
\end{corollary}

\begin{proof}
If $M$ is surely optimal as stated, we have by the ``if'' part of Theorem~\ref{surechar} (see (\ref{unrep})),%
\begin{equation}
Y_{i}^{\ast}-Y_{0}^{\ast}-M_{i}=-\sum_{l=1}^{i}\left(  A_{l}^{\ast}%
-A_{l-1}^{\ast}\right)  \zeta_{i}=-N_{i}, \label{ncor}
\end{equation}
with $N$ being adapted, non-decreasing and $N_{0}=0.$ Conversely, if
\[
Y_{i}^{\ast}=Y_{0}^{\ast}+M_{i}-N_{i}%
\]
for some martingale $M,$ $M_0=0,$ and non-decreasing adapted $N,$
$N_0=0,$ we consider for each $i,$ $0\leq i\leq T,$%
\[
\max_{i\leq j\leq T}\left(  Z_{j}-M_{j}+M_{i}\right)  =\max_{i\leq j\leq
T}\left(  Z_{j}-Y_{j}^{\ast}-N_{j}+Y_{i}^{\ast}+N_{i}\right)  \leq Y_{i},%
^{\ast}%
\]
and then apply (\ref{uni}) again.
\end{proof}

\medskip

We have the following remark.

\begin{remark}\label{impalg}
Let the martingale $M$ with $M_{0}=0$ be surely optimal for $i=0,\ldots,T.$
For the non-decreasing process $N$ defined by  (\ref{ncor}) it holds that%
\[
Y_{i}^{\ast}-M_{i}+M_{i-1}-Z_{i-1}=Y_{i-1}^{\ast}-N_{i}+N_{i-1}-Z_{i-1}%
=:U_{i},
\]
and since by  (\ref{ncor}), $N_{i}-N_{i-1}=\left(  A_{i}^{\ast}-A_{i-1}^{\ast
}\right)  \zeta_{i}$, we obtain from \eqref{df}%
\[
\left(  U_{i}\right)  ^{+}=Y_{i-1}^{\ast}-Z_{i-1}\text{ \ \ a.s.}%
\]
So, in particular we have that $\left(  U_{i}\right)  ^{+}$ is $\mathcal{F}%
_{i-1}$-measurable while $U_{i}$ itself is generally \textbf{not}, except for
the case where $M=M^{\ast}.$ A similar observation will encountered later on in \eqref{eq:was01}.
\end{remark}

From Theorem~\ref{surechar} it is clear that there exist infinitely many martingales
which are surely optimal for all $i=0,\ldots,T.$ In the following example we construct a one-parametric
family of such martingales which includes the Doob martingale of the Snell envelope.
\begin{example}
Let us assume $Z>0$ a.s. (if $Z$ is strictly bounded from below by a constant
$-K,$ we may consider the equivalent stopping problem due to $Z+K$). Then
$Y^{\ast}>0$ a.s., and for any $\alpha,$ $0\leq\alpha\leq1,$ we consider
\[
\zeta_{i}:=1-\alpha+\alpha\frac{Y_{i}^{\ast}}{E_{i-1}Y_{i}^{\ast}}%
=1-\alpha+\alpha\frac{N_{l}^{\ast}}{N_{l-1}^{\ast}},
\]
where $N^{\ast}$ is the martingale part of the multiplicative decomposition
$Y_{i}^{\ast}=Y_{0}^{\ast}N_{i}^{\ast}B_{i}^{\ast}$ of the Snell envelope (see
Examples~\ref{ex}). Obviously, it holds $E_{i-1}\zeta_{i}=1$ and $\zeta_{i}\geq0,$
and hence, by Theorem~\ref{surechar} we obtain for every $0\leq\alpha\leq1$ a
martingale%
\begin{align*}
M_{i}  & =M_{i}^{\ast}-A_{i}^{\ast}+\sum_{l=1}^{i}\left(  A_{l}^{\ast}%
-A_{l-1}^{\ast}\right)  \left(  1-\alpha+\alpha\frac{N_{l}^{\ast}}%
{N_{l-1}^{\ast}}\right)  \\
& =M_{i}^{\ast}-\alpha A_{i}^{\ast}+\alpha\sum_{l=1}^{i}\left(  A_{l}^{\ast
}-A_{l-1}^{\ast}\right)  \frac{N_{l}^{\ast}}{N_{l-1}^{\ast}},
\end{align*}
which is surely optimal for $i=0,...,T.$ Thus, for $\alpha=0$ (i.e.
$\zeta_{i}\equiv1$) we retrieve the standard Doob martingale of the Snell
envelope, and for $\alpha=1$ we obtain
\begin{align}
M_{i}  & =Y_{i}^{\ast}-Y_{0}^{\ast}+\sum_{l=1}^{i}\left(  A_{l}^{\ast}%
-A_{l-1}^{\ast}\right)  \frac{N_{l}^{\ast}}{N_{l-1}^{\ast}}\nonumber\\
& =\sum_{l=1}^{i}\left(  Y_{l}^{\ast}-Y_{l-1}^{\ast}+Y_{l-1}^{\ast}\left(
1-\frac{B_{l}^{\ast}}{B_{l-1}^{\ast}}\right)  \frac{N_{l}^{\ast}}%
{N_{l-1}^{\ast}}\right)  \nonumber\\
& =Y_{0}^{\ast}\sum_{l=1}^{i}\left(  N_{l}^{\ast}B_{l}^{\ast}-N_{l-1}^{\ast
}B_{l-1}^{\ast}+B_{l-1}^{\ast}\left(  1-\frac{B_{l}^{\ast}}{B_{l-1}^{\ast}%
}\right)  N_{l}^{\ast}\right)  \nonumber\\
& =Y_{0}^{\ast}\sum_{l=1}^{i}B_{l-1}^{\ast}\left(  N_{l}^{\ast}-N_{l-1}^{\ast
}\right)  .\label{newm}%
\end{align}
Note that this martingale differs from the martingale $Y_{0}^{^{\ast}}%
(N_{i}^{\ast}-1)$ from Example~\ref{ex} (they would coincide after dropping the factors
$B_{l-1}^{\ast}$). It is easy to show (using Theorem~\ref{surechar}
again) that the latter martingale is in general only optimal at $i=0,$ while
the martingale (\ref{newm}) is surely optimal for all $i=0,...,T,$ by construction.
\end{example}
The next theorem provides a key criterion for identifying
surely optimal martingales.
\begin{theorem}
\label{sureth} Let $Y^{\ast}$ be the Snell envelope of the cash-flow $Z$ and
let $M$ be any martingale. Then, for any $i\in\{0,...,T\}$ it holds
\[
\max_{i\leq j\leq T}\left(  Z_{j}-M_{j}+M_{i}\right)
\in\mathcal{F}_{i}\quad\Longrightarrow\quad \max_{i\leq
j\leq T}\left(  Z_{j}-M_{j}+M_{i}\right)  =Y_{i}^{\ast}.
\]

\end{theorem}

\begin{proof}
Let us suppose
$\vartheta_{i}:=\max_{i\leq
j\leq T}\left(  Z_{j}-M_{j}+M_{i}\right) \in\mathcal{F}_{i}$  and define the stopping time%
\begin{align*}
\tau_{i}  & =\inf\left\{  j\geq i:Z_{j}-M_{j}+M_{i}%
\geq\vartheta_{i}\right\}.
\end{align*}
By the definition of $\vartheta_{i}$ we have $i\le \tau_{i}\le T$ almost surely.  We thus have%
\begin{align*}
Y_{i}^{\ast} &  \geq E_{i}\, Z_{\tau_{i}}\geq E_{i}\left(M_{\tau_{i}}-M_{i}+\vartheta_{i}\right)
= \vartheta_{i},
\end{align*}
by Doob's optional sampling theorem and the fact that $\vartheta_{i}\in\mathcal{F}_{i}.$
On the other hand we have $\vartheta_{i}$ $=$ $E_{i}\,\vartheta_{i}
\geq%
Y_{i}^{\ast}$ due to (\ref{uni}).
\end{proof}

\begin{remark}\label{surecont}
While in this paper we work in a  discrete time setting, it is obvious that Theorem~\ref{sureth}  can be proved in (almost) literally the same way for continuous time exercise as well.
\end{remark}

\section{Stability of surely optimal martingales}

\label{secnear} In equivalent terms, Theorem~\ref{sureth} states that, if a
martingale $M$ is such that the conditional variance of
\[
\vartheta_{i}{(M)}:=\max_{i\leq j\leq T}\left(  Z_{j}-M_{j}+M_{i}\right)
\]
is zero for some $0\le i\le T$, i.e.%
\[
\Var_{i}\,  \vartheta_{i}{(M)}  :=E_{i}\,(\vartheta_{i}(M)%
-E_{i}\vartheta_{i}(M))^{2}=0,\text{ a.s.,}
\]
then $\vartheta_{i}{(M)}=Y_{i}^{\ast}.$  Hence the martingale
$M$ is surely optimal at $i.$ In this section we present a stability
result for martingales $M$ which are, loosely speaking, close to be surely
optimal at some $i,$ in the sense that \Var$_{i}\,  \vartheta_{i}%
{(M)}  $ is small. More specifically, we provide mild conditions on a  sequence of martingales
$(M^{(n)})_{n\geq1}$  which
guarantee  that the
corresponding upper bounds converge to the Snell envelope in a sense, although the
sequence of martingales $(M^{(n)})$ does not necessarily converge. We have
the following result.

\begin{theorem}
\label{near} \label{stab} Let $i\in\{0,...,T\}$.
If \Var$_{i}\, \vartheta_{i}^{(n)}\overset{P}{\rightarrow}0$ for
$n\rightarrow\infty,$ where $\vartheta_{i}^{(n)}:=\vartheta_{i}(M^{(n)}),$ and if in addition the sequence
 of martingales $\left(  M_{i}^{(n)}\right)  _{n\geq1}$ is
uniformly integrable, then it holds
\[
E_{i}\, \vartheta_{i}^{(n)}\overset{L_{1}}{\rightarrow}Y_{i}^{\ast}.
\]
\end{theorem}

\begin{proof}
Fix an $i\in\{0,...,T\}$ and suppose that the assumptions of the theorem are
satisfied. Now take an $\epsilon>0.$ By introducing an auxiliary time
$\partial>T$ and setting $Z_{\partial}=0$ we next define the stopping time%
\[
\tau_{i}^{(n)}=\inf\left\{  j\geq i:Z_{j}-M_{j}^{(n)}+M_{i}^{(n)}\geq
E_{i}\vartheta_{i}^{(n)}-\epsilon\right\}  \wedge\partial.
\]
We thus have with $M_{\partial}^{(n)}:=M_{T}^{(n)},$ $n\geq1,$
\begin{align*}
Y_{i}^{\ast} &  \geq E_{i}\,Z_{\tau_{i}^{(n)}}=E_{i}\,Z_{\tau_{i}^{(n)}%
}1_{\tau_{i}^{(n)}<\partial}\geq E_{i}\,\left(  M_{\tau_{i}^{(n)}}^{(n)}%
-M_{i}^{(n)}+E_{i}\vartheta_{i}^{(n)}-\epsilon\right)  1_{\{\tau_{i}%
^{(n)}<\partial\}}\\
&  \!\!\!=E_{i}\left(  M_{\tau_{i}^{(n)}}^{(n)}-M_{i}^{(n)}+E_{i}%
\vartheta^{(n)}-\epsilon\right)  -E_{i}\left(  M_{T}^{(n)}-M_{i}^{(n)}%
+E_{i}\vartheta_{i}^{(n)}-\epsilon\right)  1_{\{\tau_{i}^{(n)}=\partial\}}\\
&  \!\!\!=E_{i}\vartheta_{i}^{(n)}-\epsilon-E_{i}\left(  M_{T}^{(n)}%
-M_{i}^{(n)}+E_{i}\vartheta_{i}^{(n)}-\epsilon\right)  1_{\{\tau_{i}%
^{(n)}=\partial\}}\text{ \ \ \  a.s.,}%
\end{align*}
hence%
\begin{align}
E_{i}\vartheta_{i}^{(n)} &  \leq Y_{i}^{\ast}+\epsilon+E_{i}\left\vert
M_{T}^{(n)}-M_{i}^{(n)}+E_{i}\vartheta_{i}^{(n)}-\epsilon\right\vert
1_{\tau_{i}^{(n)}=\partial}\nonumber\\
&  =:Y_{i}^{\ast}+\epsilon+E_{i}\,U_{i}^{(n)}1_{\tau_{i}^{(n)}=\partial}\text{
\ \ a.s.}\label{low}%
\end{align}
Now it is easy to see that the family of random variables $\left(  U_{i}%
^{(n)}\right)  _{n\geq1}$ is uniformly integrable too. We so may take
$K_{\epsilon}>0$ such that%
\[
\sup_{n\geq0}E\,U_{i}^{(n)}1_{U_{i}^{(n)}>K_{\epsilon}}\leq\epsilon.
\]
Further observe that due to a conditional version of Chebyshev's inequality,
\[
0\leq E_{i}\,1_{\left\{  \tau_{i}^{(n)}=\partial\right\}  }=E_{i}1_{\left\{
\vartheta_{i}^{(n)}<E_{i}\vartheta_{i}^{(n)}-\epsilon\right\}  }\leq
\frac{{\Var}_i\,\vartheta_{i}^{(n)}}{\epsilon^{2}}\overset{P}{\rightarrow}0.
\]
Since the family $\left(  E_{i}\,1_{\left\{  \tau_{i}^{(n)}=\partial\right\}
}\right)  _{n\geq0}$ is bounded by $1$, it is uniformly integrable.
Hence, it follows that%
\begin{equation}
\,E_{i}1_{\left\{  \tau_{i}^{(n)}=\partial\right\}  }\overset{L_{1}%
}{\rightarrow}0.\label{convl1}%
\end{equation}
We thus have%
\begin{align*}
E\,U_{i}^{(n)}1_{\tau_{i}^{(n)}=\partial} &  =E\,U_{i}^{(n)}1_{U_{i}%
^{(n)}>K_{\epsilon}}1_{\tau_{i}^{(n)}=\partial}+E\,U_{i}^{(n)}1_{U_{i}%
^{(n)}\leq K_{\epsilon}}1_{\tau_{i}^{(n)}=\partial}\\
&  \leq\epsilon+K_{\epsilon}E\,1_{U_{i}^{(n)}\leq K_{\epsilon}}1_{\tau
_{i}^{(n)}=\partial}\leq\epsilon+K_{\epsilon}E\,E_{i}1_{\tau_{i}%
^{(n)}=\partial}<2\epsilon
\end{align*}
for $n>N_{\epsilon,K_{\epsilon}}$ by (\ref{convl1}). So for $n>N_{\epsilon
,K_{\epsilon}},$ we derive from (\ref{low})
\[
E\vartheta_{i}^{(n)}\leq EY_{i}+\epsilon+E\,U_{i}^{(n)}1_{\tau_{i}%
^{(n)}=\partial}\leq EY_{i}^{\ast}+3\epsilon.
\]
Thus,%
\[
\overline{\lim}_{n\rightarrow\infty}\text{ }E\vartheta_{i}^{(n)}\leq
EY_{i}^{\ast}+3\epsilon
\]
Since $\epsilon>0$ was arbitrary,%
\[
\overline{\lim}_{n\rightarrow\infty}\text{ }E\vartheta_{i}^{(n)}\leq
EY_{i}^{\ast}.
\]
On the other hand, due to (\ref{uni}) we have $E_{i}\vartheta_{i}^{(n)}\geq
Y_{i}^{\ast}$ a.s. for all $n,$ so
\[
0\leq\overline{\lim}_{n\rightarrow\infty}\,E\left\vert E_{i}\vartheta
^{(n)}-Y_{i}^{\ast}\right\vert =\overline{\lim}_{n\rightarrow\infty}\left(
E\vartheta_{i}^{(n)}-EY_{i}^{\ast}\right)  \leq0,
\]
which finally proves $E_{i}\vartheta_{i}^{(n)}\overset{L_{1}}{\rightarrow}Y_{i}^{\ast}.$
\end{proof}

\begin{remark}
Like Theorem~\ref{sureth}, Theorem~\ref{near} can be formulated in a continuous time setting as well with (almost) literally the same proof.
\end{remark}
The following simple example illustrates that Theorem~\ref{stab} would not be
true when the uniform integrability condition is dropped.

\begin{example}
\label{counter1} Take $T=1,$ $Z_{0}=Z_{1}=0,$ $M_0^{(n)}=0,$ $M_{1}^{(n)}=:-\xi_{n}$ with
$E_{0}\xi_{n}=0,$ $n=1,2,\ldots$ Then obviously $Y^*_0=0,$ and we have%
\[
\vartheta_{0}^{(n)}=\max(Z_{0}-M_{0}^{(n)},Z_{1}-M_{1}^{(n)})=\max(0,\xi
^{(n)})=\xi_{+}^{(n)}.%
\]
Now take%
\[
\xi^{(n)}=\left\{
\begin{tabular}
[c]{l}%
$1$ with Prob. $\frac{n-1}{n}$\\
$1-n$ with Prob. $\frac{1}{n}$%
\end{tabular}
\ \ \ \ \ \right.
\]
(hence $E_{0}\xi^{(n)}=0$). Then, for $n\rightarrow\infty$ we have
$\Var_{0}\, \vartheta_{0}^{(n)}
=E_{0}\, (\xi_{+}^{(n)})^{2}-\left(
E_{0}\xi_{+}^{(n)}\right)  ^{2}=\frac{n-1}{n}-\left(  \frac{n-1}{n}\right)
^{2}=\frac{n-1}{n^{2}}\rightarrow0,$ whereas $E_{0}\vartheta_{0}^{(n)}%
=E_{0}\xi_{n}^{+}=\frac{n-1}{n}\rightarrow1.$ Clearly, for
each $K>1,$  $E_{0}%
\,\left\vert M_{1}^{(n)}\right\vert 1_{\{\left\vert M_{1}^{(n)}\right\vert
>K\}}\geq\frac{n-1}{n}1_{\{n-1>K\}}\rightarrow1$ as $n\rightarrow\infty,$ hence the
$\left(  M_{1}^{(n)}\right)  $ are not uniformly integrable.
\end{example}
In view of the next Corollary,
Theorem~\ref{near} may be considered as a
stability theorem related to Theorem~\ref{sureth}.

\begin{corollary}
\label{stab1} Let $\mathcal{M}^{UI}$ be a set of uniformly integrable
martingales. Then for any $i\in\{0,\ldots,T\}$ it holds: For every  $\epsilon>0$ there exist a
$\delta>0$ such that
\[
\left[  M\in\mathcal{M}^{UI}\text{ \ \ and \ \ }E\,\Var_{i}\,\vartheta
_{i}(M)<\delta\right]  \text{ \ \ }\Longrightarrow\text{ \ }0\leq
E\,\vartheta_{i}(M)-Y_{i}^{\ast}<\epsilon.
\]

\end{corollary}

\begin{proof}
Suppose the statement is not true for some $i.$ Then there exists an
$\epsilon_{0}>0$ such that for all $n\in\mathbb{N}$ there exists a martingale
$M^{(n)}\in\mathcal{M}^{UI},$  for which $E\,\Var_{i}\,\vartheta_{i}%
(M^{(n)})<1/n$ and $E\,\vartheta_{i}(M^{(n)})-Y_{i}^{\ast}\geq$ $\epsilon
_{0}.$ Since convergence in $L_{1}$ implies convergence in probability along a subsequence (indexed again by $n$) we thus
have $\Var_{i}\,\vartheta_{i}(M^{(n)})\overset{P}{\rightarrow}0,$ and
$E\,\left\vert \vartheta_{i}(M^{(n)})-Y_{i}^{\ast}\right\vert \geq$
$\epsilon_{0}$ along this subsequence. This contradicts Theorem \ref{near}.
\end{proof}

\begin{remark}\label{wien}
Theorem~\ref{near} is important in practical situations, for
instance, for (possibly high dimensional) underlyings of jump-diffusion type in a L\'evy-It\^o setup. In this
environment we may consider the following class of uniformly integrable martingales.

\medskip

\noindent Let $W$ be an $m$-dimensional Brownian motion and let $N$ denote a Poisson random measure, independent of $W$, with (deterministic) compensator measure $\nu(s,du)ds$ such that
\begin{align*}
\int_0^t \int_{\mathbb{R}^q} (u^2 \wedge 1) \nu(s,du)ds < \infty, ~ 0 \leq t \leq T.
\end{align*}
Let $(\mathcal{F}_t)_{0 \leq t \leq T}$ be the filtration generated by $W$ and $N$, augmented by null sets. Now let $X$ be a $D$-dimensional Markov process, adapted to $(\mathcal{F}_t)$, and consider the mappings $c: [0,T] \times \mathbb{R}^D \to \mathbb{R}_{\geq 0}$ and $d: [0,T] \times \mathbb{R}^D \times \mathbb{R}^q \to \mathbb{R}_{\geq 0}$ satisfying
\begin{align}\label{eq:dummy01}
&E \int_0^T |c(s,X_s)|^2 ds < \infty, \quad E \int_0^T \int_{\mathbb{R}^q} |d(s,X_{s},u)|^2 \nu(s,du)ds< \infty.
\end{align}
\noindent We define the class  of uniformly integrable martingales, $\mathcal{M}^{UI}$, as the set of all martingales $M$ satisfying
\begin{align*}
M_t &= M_0 + M^c_t + M^d_t\\
&= M_0 + \int_0^t \varphi^c(s,X_s) dW_s + \int_0^t \int_{\mathbb{R}^q} \varphi^d(s,X_{s},u) \tilde{N}(ds,du),
\end{align*}
where $\varphi^c$ and $\varphi^d$ satisfy
\begin{align*}
|\varphi^c| \leq c, \quad  |\varphi^d| \leq d,
\end{align*}
and $\tilde{N} = N -\nu$ is the compensated Poisson measure.
Note that $M$ is indeed a martingale and that  the expected quadratic variation of $M$ is given by
\begin{align*}
E\,\big[ M,M \big]_t &= E\int_0^t |\varphi^c(s,X_s)|^2 ds + E\int_0^t \int_{\mathbb{R}^q} |\varphi^d(s,X_s,u)|^2 \nu(s,du)ds \\
&\leq E\int_0^t |c(s,X_s)|^2 ds + E\int_0^t \int_{\mathbb{R}^q} |d(s,X_s,u)|^2 \nu(s,du)ds.
\end{align*}
We then have for every $t\in [0,T]$ ,
\begin{align*}
\sup_{M \in \mathcal{M}^{UI}} E |M_t|^2 &\leq \sup_{M \in \mathcal{M}^{UI}} E \sup_{0\leq t \leq T}|M_t|^2 \leq \sup_{M \in \mathcal{M}^{UI}} C E \big[ M,M \big]_T < \infty,
\end{align*}
where the second estimation results from the Burkholder-Davis-Gundy inequality and the third estimation follows from \eqref{eq:dummy01}. Finally, an application of the de la Vall\'ee Poussin criterion yields that $\mathcal{M}^{UI}$ is indeed a family of uniformly integrable martingales.
\end{remark}

\section{New dual algorithms for pricing of \newline Bermudan
derivatives}\label{secalg}

In this section we consider the design of a new dual algorithm for solving
multiple stopping problems, hence pricing Bermudan products, which are based on the
theoretical insights from Theorem~\ref{sureth}, Theorem~\ref{near}, and Corollary~\ref{stab1}. In the following, we provide an
assessment of the merits of variance minimizing dual algorithms based on these results.

\subsection{Merits of variance minimizing dual algorithms}\label{secalg0}

Let $Q$ be some index set and $\mathcal{M}=\{M^{q}:q\in Q\}$
be a set of uniformly integrable martingales such that $\mathcal{M}$ contains
a martingale $M^{q^{\ast}}$ which is surely optimal at $i=0.$ Suppose that for
any $q\in Q$ we have $N$ samples of $\vartheta_{0}(M^{q,(n)}),$
$n=1,...,N.$ Based on these samples we may estimate $\Var_{0}\,\vartheta_{0}%
(M^{q})=\Var\,\vartheta_{0}(M^{q})$ as usual by%
\begin{align}
\Var^{(N)}\,\vartheta_{0}(M^{q}) &  :=\frac{1}{N-1}\sum_{n=1}^{N}\left(
\vartheta_{0}(M^{q,(n)})-\overline{\vartheta_{0}(M^{q})_N}\right)
^{2},\text{ \ \ with}\nonumber\\
\overline{\vartheta_{0}(M^{q})_N} &  :=\frac{1}{N}\sum_{n=1}^{N}%
\vartheta_{0}(M^{q,(n)}).\label{VE}%
\end{align}
So, in principle, only two realizations ($N=2$) would be enough to identify a
$q^{\ast}$ such that
\[
0=\Var\,\vartheta_{0}(M^{q^{\ast}})=\Var^{(2)}\,\vartheta_{0}(M^{q^{\ast}}%
)=\min_{q\in Q}\sum_{n=1}^{2}\left(  \vartheta_{0}(M^{q,(n)})-\overline
{\vartheta_{0}(M^{q})_N}\right)  ^{2},
\]
and then obtain $Y_{0}^{\ast}=$ $\vartheta_{0}(M^{q^{\ast}})=\vartheta
_{0}(M^{q^{\ast},(1)}).$ Due to this stylized argumentation we may expect that
in a case where although the set $\mathcal{M}$ doesn't contain a martingale
that is surely optimal at $i=0$ but at least one martingale $M^{q}$ such that
$\Var\,\vartheta_{0}(M^{q})$ is \textquotedblleft small
enough\textquotedblright, we only need a relatively small sample size $N$ to
identify this martingale, leading to a tight upper bound $Y_0^{up}:=E\vartheta_{0}%
(M^{q}).$ In the following, we formalize this idea by giving precise estimates for the variance estimators in terms of the family of uniformly integrable martingales.

\medskip

Suppose we want to obtain an upper bound which is bounded from above by $Y_{0}^{\ast}%
+\epsilon$ for some given $\epsilon>0.$ Consider a family of uniformly integrable martingales $\mathcal{M}=\{M^{q}:q\in Q\}$ which is \textquotedblleft rich
enough\textquotedblright\ in the sense that there exists a $\delta>0$
according to Corollary~\ref{stab1}, such that
\[
\{M\in\mathcal{M}: \Var\,\vartheta_{0}(M)<\kappa\delta
\}\neq\varnothing,
\]
for some $0<\kappa<1.$ Then, in particular, we have $\inf_{q\in Q}\Var\,\vartheta
_{0}(M^{q})<\kappa\delta$ which implies that for any $q \in Q$ we have $E \Var\,\vartheta
_{0}(M^{q}) = \Var\,\vartheta
_{0}(M^{q})<\delta$. Now Corollary~\ref{stab1} again yields $0\leq E\vartheta_{0}(M^{q})-Y_{0}^{\ast
}<\epsilon$. For a given set of realizations
$\vartheta_{0}(M^{q,(n)}),$ $n=1,...,N,$ $q\in Q,$ we now may try to find an optimal $q_{N}^{\circ}\in Q$ arising as the solution of the minimization problem%
\begin{equation}
\inf_{q\in Q}\Var^{(N)}\,\vartheta_{0}(M^{q})= \Var^{(N)}\,\vartheta_{0}%
(M^{q_{N}^{\circ}}). \label{Varmin}
\end{equation}
For convenience, we assume that such a $q_{N}^{\circ}$ does exist. Furthermore, we assume the existence of
$q^{\circ}\in Q$ which satisfies
$$
 \inf_{q\in Q}\Var\,\vartheta_{0}(M^{q})=\Var\,\vartheta_{0}(M^{q^{\circ}}).
$$
Let $0\le \alpha\ll 1$ be a small threshold probability.
Then one can show under mild conditions on the family of random variables
$\{\vartheta_{0}(M^{q}):q\in Q\},$
that for some constant $C>0$ and quantile coefficient $c_{\alpha}>0$ (only depending on $\alpha$) we have with
probability larger than $1-\alpha,$ %
\begin{eqnarray}
\Var^{(N)}\,\vartheta_{0}(M^{q})  & \leq & \Var\,\vartheta_{0}(M^{q})\left(
1+c_{\alpha}\sqrt{\frac{C}{N}}\right)  ,\label{esti}\\
\left(
1-c_{\alpha}\sqrt{\frac{C}{N}}\right)\Var\,\vartheta_{0}(M^{q})  & \leq& \Var^{(N)}\,\vartheta_{0}(M^{q}),\quad  \text{for}\quad q\in\{q^\circ, q^{\circ}_N\}\notag
\end{eqnarray}
(see Appendix for details). Thus, with probability larger than $1-\alpha,$%
\begin{eqnarray*}
\left(
1-c_{\alpha}\sqrt{\frac{C}{N}}\right)\Var\,\vartheta_{0}(M^{q_{N}^{\circ}})&\leq& \Var^{N}\,\vartheta_{0}(M^{q_N^{\circ}%
})\le  \Var^{N}\,\vartheta_{0}(M^{q^{\circ}%
})\\
&\leq& \Var\,\vartheta_{0}(M^{q^{\circ}})\left(  1+c_{\alpha}\sqrt{\frac
{C}{N}}\right),
\end{eqnarray*}
which yields
$$
\Var\,\vartheta_{0}(M^{q_{N}^{\circ}})\leq\kappa\delta\frac{1+c_{\alpha}\sqrt{\frac{C}{N}%
}}{1-c_{\alpha}\sqrt{\frac{C}{N}%
}}.
$$
This implies that for
\begin{equation}\label{SaS}
N=\inf\left\{n:\kappa\frac{1+c_{\alpha}\sqrt{\frac{C}{n}%
}}{1-c_{\alpha}\sqrt{\frac{C}{n}%
}}  <1\right\}
\end{equation}
we have $P\left[\Var\,\vartheta_{0}(M^{q_{N}^{\circ}})<\delta\right]>1-\alpha.$ Thus, by Corollary~\ref{stab1}, the dual upper bound
due to the  martingale $M^{q_{N}^{\circ}}$ identified by \eqref{Varmin} falls below $Y^*_0+\epsilon$ with probability larger than $1-\alpha.$  As a main feature, equation \eqref{SaS} demonstrates that the smaller $\kappa,$ the fewer samples $N$ we may choose for the identification of $q_{N}^{\circ}$.

\medskip

The above argumentation suggests to minimize the estimated variance of the dual estimator over a parametric set of martingales using a
relatively small sample size $N.$ However, as the parametric set of martingales $\mathcal{M}$ needs to be ``rich enough'', in practice there may be many parameters involved, which in turn may lead to a non-convex minimization problem with many
local minima. As a remedy to this problem,  rather than directly minimizing the variance of the dual estimator at time zero, we propose to
minimize backwardly the expected conditional variances $E\,\Var_i\, \vartheta_{i}(M^q)$ over $q\in Q,$ starting from $i=T$ (where
the conditional variance is trivially zero) down to $i=0,$ using a simple but effective recursive relationship between $\vartheta_{i}(M^q)$
and $\vartheta_{i+1}(M^q)$ as explained in the next subsection. For this backward minimization procedure the arguments above apply as well and moreover, as we will see, it opens the possibility for linear regression, hence also  the possibility for fast numerical implementations.

\subsection{Backward dual variance minimization}

\label{secalg1}

Motivated by Section~\ref{secalg0} we now develop a backward recursive simulation
based algorithm for the construction of a dual martingale $M$ that yields tight
upper bounds. In view of a such a Monte Carlo approach, we assume a Markov setting generated by some underlying Markov process
$X:=\left(  X_{t}\right)  _{0\leq t\leq T},$ and a cash-flow of the form
$Z_{j}:=Z_{j}(X_{j}):= Z(j,X_j).$ First we describe the algorithm in a pseudo language which
involves terms such as conditional expectations and conditional variances. Then, we spell out an implementable Monte Carlo algorithm where these expressions are replaced by their
empirical counterparts.

\medskip

To start out  on a pseudo algorithmic level we
construct a martingale $M$ backwardly in a recursive way by establishing that from $i=T$ down to $i=0$ the
expected conditional variances $E\,$Var$_{i}\vartheta_{i}(M)$ are
\textquotedblleft as small as possible\textquotedblright\ in a sense that we will describe. The
martingale $M$ is such that for $j>i$, any increment%
\begin{equation}
M_{j}-M_{i}\text{ is\ measurable with respect to }\Delta
\mathcal{F}_{i,j}:=\sigma\{X_{s}:{i}\leq s\leq {j}\}.\label{meas}%
\end{equation}
It is easy to see that the Doob martingale of the Snell envelope
meets this measurability property, however, in general Theorem~\ref{surechar} yields that there may exist many other surely optimal martingales satisfying this property.

\medskip

A corner stone of the whole procedure is the following recursion
that holds for any martingale $M$ and any $i<T,$%
\begin{align}
\vartheta_{i}(M) &  =\max\left(  Z_{i},\max_{i+1\leq j\leq T}\left(
Z_{j}-M_{j}+M_{i}\right)  \right) \nonumber\\
&  =\max\left(  Z_{i},\vartheta_{i+1}(M)+M_{i}-M_{i+1}\right)  \nonumber\\
&  =Z_{i}+\left(  \vartheta_{i+1}(M)+M_{i}-M_{i+1}-Z_{i}\right)  ^{+}.\label{eq:was01}
\end{align}
Obviously, at every $i=0,...,T,$ $\vartheta_{i}(M)$ only depends on $\left(
M_{j}-M_{i}\right)  _{i\leq j\leq T},$ and at the starting time $i=T$ we
initially have $\vartheta_{T}(M)=Z_{T}$ which trivially satisfies
$E\,$Var$_{T}\left(  \vartheta_{T}(M)\right)  =0.$ Note that if $M$ were already surely optimal, i.e. $\vartheta(M)$ were already equal to $Y^\ast$, then Remark \ref{impalg} would imply that for $U_i := \vartheta_{i+1}(M)+M_{i}-M_{i+1}-Z_{i}$, $(U_i)^+ = \vartheta_{i}(M) - Z_i = Y^\ast_i - Z_i$ is already $\mathcal{F}_i$-measurable.

\medskip

Now the essential idea is comprised in the following backward induction: Assume that for $i+1\leq T$ we
have constructed the increments $\left(  M_{j}-M_{i+1}\right)
_{i+1\leq j\leq T}$ and $\vartheta_{i+1}(M).$ Now the task is to find a random variable $\xi_{i+1}$ such
that%
\begin{equation}
\xi_{i+1}\text{ is }\Delta\mathcal{F}_{i,i+1}\text{-measurable, \ \ \ \ }%
E_{i}\xi_{i+1}=0,\label{coxi}%
\end{equation}
that solves the following minimization problem
\begin{eqnarray}
\xi_{i+1}&:=&\underset{\xi\in\Delta\mathcal{F}_{i,i+1},\text{ }E_{i}\xi=0}%
{\arg\min}E\,\Var_{i}\vartheta_{i}(M(\xi))\notag \\
&=&\underset{\xi\in\Delta
\mathcal{F}_{i,i+1},\text{ }E_{i}\xi=0}{\arg\min}E\,\Var_{i}\left(
\vartheta_{i+1}(M)-\xi-Z_{i}\right)  ^{+}.\label{vm}%
\end{eqnarray}
Intuitively, $\xi_{i+1}$ represents the optimal martingale increment and thus, we put $M_{j}(\xi_{i+1})-M_{i}(\xi_{i+1}):=M_{j}-M_{i+1}+\xi_{i+1}$ for $j\geq i+1$. By construction, the random variable $\xi_{i+1}$
satisfies (\ref{coxi}), therefore, we obtain a set of martingale increments
$\left(  M_{j}(\xi_{i+1})-M_{i}(\xi_{i+1})\right)  _{i\leq j\leq T}$, which has now been extended from $j=i+1$ to $j=i$ and  which
satisfies for $j\geq i+1,$
\begin{align*}
M_{j}(\xi_{i+1})-M_{i+1}(\xi_{i+1})&=M_{j}(\xi_{i+1})-M_{i}(\xi_{i+1}%
)+M_{i}(\xi_{i+1})-M_{i+1}(\xi_{i+1})\\
&=M_{j}-M_{i+1}
\end{align*}
and by construction also the measurability requirement (\ref{meas}). Now we extend the increments $\left(  M_{j}-M_{i+1}\right)
_{i+1\leq j\leq T}$ from $j=i+1$ to $j=i$ by setting
$$\left(  M_{j}-M_{i}\right)_{i\leq j\leq T} = \left(  M_{j}(\xi_{i+1})-M_{i}(\xi_{i+1})\right)  _{i\leq j\leq T}$$
Finally we put
\[
\vartheta_{i}(M)=Z_{i}+\left(  \vartheta_{i+1}(M)-\xi_{i+1}-Z_{i}\right)
^{+}.
\]
After carrying out these steps backwardly from $i=T$ down to $i=0$ we
end up with a family of martingale increments $\left(  M_{j}%
-M_{0}\right)  _{0\leq j\leq T},$ hence a martingale $\left(  M_{j}\right)
_{0\leq j\leq T}$, as $M_{0}=0$ without loss of generality. This
martingale will be subsequently used to compute a dual upper bound for
$Y_{0}^{\ast}$ via
\begin{align*}
Y^{up}_0 = E \max_{0 \leq j \leq T} \big( Z_j - M_j \big).
\end{align*}

\medskip

The key step in the above procedure is to find a solution to minimization problem \eqref{vm}: Suppose that the martingale increments satisfying (\ref{coxi})
for some fixed $i$ may be parametrized as $\xi_{i+1}(\beta)$ where $\beta$ is
some generic parameter. Based on a set of simulated trajectories of $X$ one may
then estimate for some $\beta$ (which we specify in more details below) the conditional variance%
\[
E\,\Var_{i}\left(  \vartheta_{i+1}(M)-\xi_{i+1}(\beta)-Z_{i}\right)
^{+}=:E\,\Var_{{i}}U_{i}^{+}(\beta)
\]
by using e.g. kernel estimators (e.g. see \cite{Liero}), and next
minimize with respect to $\beta.$ In particular when the dimension of the
parameter space is very small (typically one-dimensional) this may lead to a feasible Monte Carlo
procedure. However, if the set of martingale increments $\xi_{i+1}(\beta)$ is
\textquotedblleft rich enough\textquotedblright\ and is moreover linearly
structured in $\beta,$ that is
\[
\xi_{i+1}(\beta)=\sum_{k=1}^{K}\beta_{k}\mathfrak{m}_{i+1}^{(k)},
\]
where $\beta = (\beta_1,\ldots,\beta_K)\in\mathbb{R}^{K}$ and the random variables $\mathfrak{m}_{i+1}^{(k)}, k=1,\ldots,K,$ satisfy (\ref{coxi}) for $K\geq 1$ sufficiently large, it is in general more effective to solve the dominating
problem%
\begin{align}
\underset{\beta\in\mathbb{R}^{K}}{\arg\min}E\,\Var_{{i}}U_{i}(\beta) &
:=\underset{\beta\in\mathbb{R}^{K}}{\arg\min}E\,\Var_{{i}}\left(  \vartheta
_{i+1}(M)-\xi_{i+1}(\beta)-Z_{i}\right)  \nonumber\\
&  =\underset{\beta\in\mathbb{R}^{K}}{\arg\min}E\,\Var_{{i}}\left(
\vartheta_{i+1}(M)-\sum_{k=1}^{K}\beta_{k}\mathfrak{m}_{i+1}^{(k)}\right)
.\label{plo1}%
\end{align}
The reason is twofold. On the one hand, if we succeed to find $\beta^{\circ
}\in\mathbb{R}^{K}$ such that $E\,\Var_{{i}}U_{i}(\beta^{\circ})$ is
sufficiently small (if it were zero, we would have arrived at an surely optimal martingale increment), then since
\[
\underset{\beta\in\mathbb{R}^{K}}{\arg\min}E\,\Var_{X_{T_{i}}}U_{i}^{+}%
(\beta)\leq E\,\Var_{i}U_{i}^{+}(\beta^{\circ})\leq E\,\Var_{{i}}U_{i}(\beta^{\circ}),
\]
$E\,\Var_{{i}}U_{i}^{+}(\beta^{\circ})$ is generally even closer to zero
and so $\beta^{\circ}$ can be considered a good approximation to (\ref{vm}) as
well. On the other hand, most importantly, problem (\ref{plo1}) can be treated
as a linear regression problem,
\begin{equation}
\left[  \beta^{\circ},\gamma^{\circ}\right]  =\underset{\beta
\in\mathbb{R}^{K},\gamma\in\mathbb{R}^{K_1}}{\arg\min}E\,\left|  \vartheta_{i+1}(M)-\sum_{k=1}^{K}%
\beta_{k}\mathfrak{m}_{i+1}^{(k)}-\sum_{k=1}^{K_{1}}\gamma_{k}\psi_{k}%
({i},X_{{i}})\right|  ^{2},\label{alin}%
\end{equation}
which employs an additional set of basis functions $\psi_{k}(t,x),$ $k=1,...,K_{1}.$
To see this, note that (\ref{alin}) is equivalent with%
\begin{align*}
\left[  \beta^{\circ},\gamma^{\circ}\right]    & =\underset{\beta
\in\mathbb{R}^{K},\gamma \in \mathbb{R}^{K_1}}{\arg\min}EE_{i}\,\left(  \vartheta_{i+1}(M)-\sum_{k=1}%
^{K}\beta_{k}\mathfrak{m}_{i+1}^{(k)}-E_{i}\vartheta_{i+1}(M)\right.  \\
& \left.  +E_{i}\vartheta_{i+1}(M)-\sum_{k=1}^{K_{1}}\gamma_{k}\psi_{k}%
({i},X_{{i}})\right)  ^{2}\\
& =\underset{\beta
\in\mathbb{R}^{K},\gamma \in \mathbb{R}^{K_1}}{\arg\min}\left\{  E\,\Var_{i}%
\left(  \vartheta_{i+1}(M)-\sum_{k=1}^{K}\beta_{k}\mathfrak{m}_{i+1}%
^{(k)}\right)  \right.  \\
& \left.  +E\,\left(  E_{i}\vartheta_{i+1}(M)-\sum_{k=1}^{K_{1}}\gamma_{k}%
\psi_{k}({i},X_{{i}})\right)  ^{2}\right\}  ,
\end{align*}
hence $\beta^{\circ}$ satisfies (\ref{plo1}), and moreover for $\gamma^{\circ
}$ it holds%
\begin{equation}
\gamma^{\circ}=\underset{\gamma\in\mathbb{R}^{K_1}}{\arg\min}\,E\,\left(
E_{i}\vartheta_{i+1}(M)-\sum_{k=1}^{K_{1}}\gamma_{k}\psi_{k}({i},X_{{i}%
})\right)  ^{2}.\label{ga0}%
\end{equation}
\ \\
Further, the regression procedure (\ref{alin}) delivers as by-product
\[
\mathcal{C}_{i}(x):=\sum_{k=1}^{K_{1}}\gamma_{k}^{\circ}\psi_{k}(i,x),
\]
an approximate continuation function that may be used afterwards to define a
stopping rule and to simulate a corresponding lower biased estimation of
$Y_{0}^{\ast}.$

\begin{remark}\label{applications}
(i) In virtually all practical applications we are in a setting as described
in Remark~\ref{wien}. In this environment we may model $\xi_{i+1}$ as linear
combinations of the form%
\begin{align}
{\xi}_{i+1}(\beta):=&\sum_{k=1}^{N_1}\beta_{k}^{c}\int_{T_{i}}^{T_{i+1}%
}\varphi_{k}^{c}(s,X_{s})dW_{s}\nonumber\\
&+\sum_{k=1}^{N_2}\beta_{k}^{d}\int_{T_{i}%
}^{T_{i+1}}\varphi_{k}^{d}(s,X_{s},u)d\widetilde{N}(ds,du),\label{wie}%
\end{align}
where $N_1+N_2=K$ and $\varphi_{k}^{c}(s,x)$ and $\varphi_{k}^{d}(s,x,u)$ are
suitable sets of basis functions satisfying the conditions in Remark~\ref{wien}. In this setting, we have
\begin{align*}
\mathfrak{m}_{i+1}^{(k)} = \int_{T_{i}}^{T_{i+1}}\varphi_{k}^{c}(s,X_{s})dW_{s} + \int_{T_{i}}^{T_{i+1}}\varphi_{k}^{d}(s,X_{s},u)d\widetilde{N}(ds,du)
\end{align*}
and $\beta = (\beta^c_1,\ldots,\beta^c_{N_1},\beta^d_1,\ldots,\beta^d_{N_2}) \in \mathbb{R}^K$.

\medskip

As an alternative, we may also take
\begin{equation}
{\xi}_{i+1}(\beta):=\sum_{k=1}^{K}\beta_{k}\left(  B_{i+1}^{(k)}%
-B_{i}^{(k)}\right)  \label{as}%
\end{equation}
for an arbitrary given set of discounted tradables $\left(  B_{j}%
^{(k)}\right)  _{0\leq j\leq T}$ where
the $B_{j}^{(k)}=B^{(k)}(j,X_{j})$ are provided by some specific problem under
consideration. For example it may happen that discounted European options
are available in closed form. In any case, (\ref{wie}) and (\ref{as}) satisfy
the requirements (\ref{coxi}) for any vector parameter $\beta\in\mathbb{R}^K.$\\
\ \\
(ii) Suppose that the system of basis martingale increments and basis
functions in the regression based minimization (\ref{alin}) is sufficiently
\textquotedblleft rich\textquotedblright\ that there even exist $\beta
^{\circ\circ}$ and $\gamma^{\circ\circ}$ such that
\[
\vartheta_{i+1}(M)-\sum_{k=1}^{K}\beta_{k}^{\circ\circ}\mathfrak{m}%
_{i+1}^{(k)}-\sum_{k=1}^{K_{1}}\gamma_{k}^{\circ\circ}\psi_{k}({i},X_{{i}})=0\text{ \ a.s.}%
\]
then one would need only one trajectory for $X$ to identify $\beta
^{\circ\circ}$ and $\gamma^{\circ\circ}$ via (\ref{alin}). This is a similar
situation as discussed in Section~\ref{secalg0}: In practice when the system (\ref{wie})
is rich enough, a relatively low sample size will be sufficient to solve
(\ref{alin}) effectively.  This phenomenon will be confirmed by our
experiments in Section~\ref{newsecnum}.
\end{remark}

\subsection*{Description of the Monte Carlo algorithm}
Let us now spell out the empirical, implementable counterpart of the procedure described above. Based on a set of trajectories $\left(  X_{j}^{(n)}\right)
_{j=0,...,T},$ $n=1,...,N,$ we carry out the following procedure.
\bigskip

\medskip

\emph{Step 1:} At $i=T$ we set on each trajectory $\vartheta_{T}^{(n)}:=\vartheta
_{T}^{(n)}(M):=Z_{T}(X_{T}^{(n)})$ and $\left(  M^{(n)}_{j} -M^{(n)}_{T}\right)  _{T\leq j\leq T} = M^{(n)}_T - M^{(n)}_T = 0$ for $n=1,...,N$.

\bigskip

\emph{Step 2:} For $n=1,...,N$ let $\left(  M^{(n)}_{j} -M^{(n)}_{i+1}\right)  _{i+1\leq j\leq T}$ be constructed. For $i=T-1$ down to $i=0$, based on the $N$ samples, we solve the
regression problem
\[
\left[  \widehat{\beta}^{(i)},\widehat{\gamma}^{(i)}\right]  :=\underset
{\beta\in\mathbb{R}^K,\gamma\in\mathbb{R}^{K_1}}{\arg\min}\frac{1}{N}\sum_{n=1}^{N}\,\left(
\vartheta_{i+1}^{(n)}(M)-\sum_{k=1}^{K}\beta_{k}\mathfrak{m}_{i+1}%
^{(k,n)}-\sum_{k=1}^{K_{1}}\gamma_{k}\psi_{k}({i},X_{{i}}^{(n)})\right)
^{2}.
\]
We then put%
\begin{align*}
&\widehat{\xi}_{i+1}^{(n)} :=\sum_{k=1}^{K}\widehat{\beta}_{k}^{(i)}%
\mathfrak{m}_{i+1}^{(k,n)},\\ 
&M^{(n)}_{j} -M^{(n)}_{i} := \left(  M^{(n)}_{j} -M^{(n)}_{i+1}\right) + \widehat{\xi}_{i+1}^{(n)},
\end{align*}
and
\[
\vartheta_{i}^{(n)}(M):=Z_{i}^{(n)}+\left(  \vartheta_{i+1}^{(n)}%
(M)-\widehat{\xi}_{i+1}^{(n)}-Z_{i}^{(n)}\right)  ^{+}.%
\]

\bigskip

\emph{Step 3:} We simulate $\widetilde{N}$ new independent samples $\left(
\widetilde{X}_{j}^{(n)}\right)  _{j=0,...,T},$ $n=1,...,\widetilde{N}$, which give rise to the new martingale samples
\[
\widetilde{M}^{(n)}_i = \sum_{j=1}^{i}\sum_{k=1}^{K}\widehat{\beta}_{k}^{(j)}\widetilde{\mathfrak{m}}_{j}^{(k,n)},\text{ \ \ }k=1,...,K,\text{
\ \ }n=1,...,\widetilde{N}.
\]

Then, an upper biased estimate for the upper bound is given by%
\begin{align}\label{eq:nix01}
\widehat{Y}_{0}^{\text{up}}:=\frac{1}{\widetilde{N}}\sum_{n=1}^{\widetilde{N}%
}\max_{0\leq i\leq T}\left(  Z_{i}^{(n)}(\widetilde{X}_{i}^{(n)})-\sum
_{j=1}^{i}\sum_{k=1}^{K}\widehat{\beta}_{k}^{(j)}\widetilde{\mathfrak{m}%
}_{j}^{(k,n)}\right)  .
\end{align}

\bigskip

\emph{Step 4:} Based on the stopping rule%
\begin{align}\label{eq:nix02}
\tau_{0}(X_i):=\inf\{i\geq 0: \text{ }Z_{i}(X_{i})\geq\sum_{k=1}^{K_{1}}%
\widehat{\gamma}_{k}^{(i)}\psi_k({i},X_{i})\}
\end{align}

we put
\[
\widehat{Y}_{0}^{\text{low}}:=\frac{1}{\widetilde{N}}\sum_{n=1}^{\widetilde
{N}}Z_{\tau_{0}(\widetilde{X}^{(n)})}^{(n)}(\widetilde{X}_{\tau_{0}%
(\widetilde{X}^{(n)})}^{(n)}).
\]
which yields lower biased estimate to $Y^\ast_0$.

\bigskip

At this point, let us briefly compare our algorithm with the algorithm from \citet{BBS}. The methodology of \citet{BBS} to compute dual martingales is built upon a procedure to numerically approximate Clark-Ocone derivatives of an approximative Snell envelope $Y$ with respect to a Wiener filtration. 
The key ingredient in \citet{BBS} is to approximate this Clark-Ocone derivative on a (fine) grid $\pi = \{t_0,\ldots,t_N \}$ which contains the exercise grid $\{0,1,\ldots,T\}$
using the estimator,
\begin{align}\label{eq:diskret}
Z^\pi_{t_j} := \frac{1}{\Delta^\pi_j}\mathbb{E}_{t_j} [\Delta^\pi W_j ~ Y_{i+1}],
\end{align}
where $\Delta^\pi_j = t_{j+1} - t_j$ and $\Delta^\pi W_j = W_{t_{j+1}} - W_{t_{j}}$. Due to \eqref{eq:diskret}, \citet{BBS} morally requires to carry out a regression at each $t_j \in \pi$ on the fine grid $\pi$. However, our algorithm only needs to carry out regressions on the coarser grid of the possible exercise dates $\{0,\ldots,T\} \subset \pi$. Moreover, note that \eqref{eq:diskret} requires as an input some approximation of $Y$, which needs to be obtained by another method, such as the method of \cite{LS}.
Furthermore, if the grid $\pi$ happens to be very fine, i.e. if $|\pi| = \sup_j |t_{j+1}-t_j| = \varepsilon $ for some very small $\varepsilon>0$, the complexity increases, and also the right-hand of \eqref{eq:diskret} becomes very large and even explodes as $\varepsilon$ approaches zero. To circumvent these instabilities, \cite{BBS} implement regressions on the coarser exercise time grid and then locally interpolate on the finer grid $\pi$. In our algorithm,  these problems do not appear at all. Finally, we underline that obtaining numerically the Clark-Ocone derivative \eqref{eq:diskret} in a non-Wiener filtration (e.g. filtratons generated by L\'evy processes) is not so straightforward. In contrast,  in our framework, the regression procedure \eqref{alin} may include jump martingales as depicted in Remark \ref{wien}.

\section{Numerical examples}\label{newsecnum}

\newcommand{\D}{\displaystyle}

In this section we present the numerical results of the backward algorithm described in Section \ref{secalg}. The performance and accuracy of our algorithm is illustrated by testing it with two benchmark examples from the literature, a Bermudan basket-put on 5 assets and Bermudan max-call on 2 and 5 assets (see \citet{BKS2} and \citet{BBS} respectively). In both examples, the risk-neutral dynamic of each asset is governed by
\[
 dX_t^d=(r-\delta)X_t^d dt+\sigma X_t^d dW_t^d,\quad d=1,...,D,
\]
where $D\in\mathbb{N}$ is the number of assets, $W_t^d$, $d=1,...,D$, are independent one-dimensional Brownian motions, and $r,\delta$ and $\sigma$ are constant real valued parameters. Exercise opportunities are equally spaced at times $T_{j}=\frac{jT}{J},\,j=0,...,J$. The discounted payoff from exercise at time $t$ is given by
\[
Z_t(X_t)=e^{-rt}(K-\frac{X_t^1+\ldots+X_t^D}{D})^{+} \quad \text{ for the Bermudan basket-put,}
\]
and
\[
Z_t(X_t)=e^{-rt}(\max(X_t^1,\ldots,X_t^D)-K)^{+} \quad \text{ for the Bermudan max-call},
\]
where we denote $X_t = (X_t^1,\ldots,X_t^D)$. For both products, the time interval $[T_j, T_{j+1}]$, $j = 0,$ $\ldots,$ $J-1$, is partitioned into $L$ equally spaced subintervals of width $\Delta t = \frac{T}{N}$ with $N = J \times L$.

\medskip

The implementation can be outlined as follows. We first simulate $M$ independent samples of  Brownian increments
$$\Delta W_i = (\Delta W_i^{1,(m)},\ldots, \Delta  W_i^{D,(m)}),\quad i = 1,\ldots, N, ~ m=1,\ldots,M.$$
Then the trajectories of $X_i^{(m)} = (X_i^{1,(m)},\ldots,X_i^{D,(m)})$, $i = 1,\ldots,N$, $m=1, \ldots,M$, are given by
\begin{equation}
X_i^{d,(m)} = X_{i-1}^{d,(m)} \exp \left( (r-\delta-\frac{1}{2} \sigma^2)\Delta t + \sigma \Delta W_i^{d,(m)}\right), \label{incr}
\end{equation}
for $d=1, \ldots,D$ and initial data $X_0=(X^1_0,\ldots,X^D_0)$.

\medskip

We now carry out the backward Monte Carlo regression algorithm as described in Section~\ref{secalg}. In this Wiener setting, we recall Remark \ref{applications} (i) and choose as the spanning family of surely optimal martingales the Wiener integrals $\mathfrak{m}_{i+1}^{(k)} = \int_{T_{i}}^{T_{i+1}}\varphi_{k}^{c}(s,X_{s})dW_{s}$. More precisely, we solve in a first step the regression problem backward in time
\begin{align}
(\widehat{\beta}^{(i)},\widehat{\gamma}^{(i)}) &  := \underset{(\beta,\gamma)}%
{\arg\min} \frac1M\sum_{m=1}^{M}\left[  \vartheta_{i+1}^{(m)}-\sum_{k=1}^{K}\beta_{k}\int_{i}^{i+1}\varphi_{k}(u,X_{u}^{(m)})dW_{u}^{(m)}\right.  \nonumber\\
&  \left.  -\sum_{k=1}^{K_{1}}\gamma_{k}\psi_{k}(i,X_{i}^{(m)})\right]  ^{2}, \quad i=T-1,\ldots,0, \label{same}
\end{align}
for two families of basis functions
$\left(  \varphi_{k}\right)=\left(  \varphi_{k}%
^{(d)}\right)$ with $\varphi_{k}%
^{(d)}=\varphi_{k}^{(1)}$, and $\left(  \psi_{k}\right)$, chosen as explained below.
In (\ref{same}) the Wiener integrals are approximated by the standard Euler scheme, using the same
Brownian increments as in (\ref{incr}). Finally, a new independent simulation is launched and we estimate an upper bound $\widehat{Y}_{0}^{up}$ and a lower bound $\widehat{Y}_{0}^{low}$ by means of \eqref{eq:nix01} and \eqref{eq:nix01}.

\medskip

As one may expect, the choice of basis functions is crucial to obtain tight upper and lower bounds. In this respect, special information on the pricing problem may help us finding suitable basis functions. One way of retrieving additional information is to employ martingales representations and Malliavian calculus techniques to obtain more specific insights into the structure of the pricing dynamics. We illustrate this by considering the following stylized setting: By the Markov property of $X$, we have that $E_t\left(Z_T(X_T)\right) = f(t,X_t)$ for some measurable function $f(t,x)$ and $0 \leq t \leq T=T_J$. Let us assume that $f(t,x)$ is differentiable in $x$. Then, by It\^o's formula and the fact that $E_t(Z_T)$ is a martingale we have
\[
Z_T(X_T) - E_{T_{J-1}}\left(Z_T(X_T)\right) = \sum_{d=1}^D \sigma \int_{T_{J-1}}^T f_{x^d} (t,X_t) X_t^d dW_t^d.
\]
Recall that $\widehat{\vartheta}_{T} =Z_T$ and $E_{T_{J-1}}\left(Z_T(X_T)\right)$ can be expressed in the following form
\[
E_{T_{J-1}}\left(Z_T(X_T)\right) = e^{-rT_{J-1}}EP(T_{J-1},X_{T_{J-1}};T),
\]
where $EP(t,x;T)$ is the price of the corresponding European option with maturity $T$ at time $t$. Thus,
it is natural to choose from time $T$ to time $T_{J-1}$  European option values for the basis
$\left(\psi_{k}(t,x)\right)$
and the corresponding European deltas multiplied by the value of the underlying asset for the basis
$\left(  \varphi_{k}(t,x)\right)$.
Although for the following steps ($t <T_{J-1}$) there is no easy way to predict optimal choices of $(\psi_{k})$ and $(\varphi_{k})$, the above analysis suggests to always include the still-alive European options into the basis $(\psi_{k})$ and include the information on the European deltas into the basis $(\varphi_{k})$.
In fact, based on similar arguments, this choice of basis functions were
already proposed in \citet{BBS}.

\subsection{Bermudan basket-put}
\label{secbsp}
In this example, we take the following parameter values,
\[
r = 0.05, \quad \delta=0,\quad \sigma =0.2,\quad D =5,\quad T =3,
\]
and
\[
X_0^1= \ldots= X_0^D = x_0,\quad K =100.
\]
We perform the simulation of the underlying asset $X$ from \eqref{incr} with a time step size $\Delta t =0.01$. For $T_j \leq t < T_{j+1}$, $j=0,\ldots,J-1$, we choose the set
$$\Big\{1, Pol_3(X_t), Pol_3(EP(t,X_t;T_{j+1})),Pol_3(EP(t,X_t;T_J))\Big\}$$
as basis functions $(\psi_{k})$, where $Pol_n(y)$ denotes the set of monomials of degree up to $n$ in the components of a vector $y$ and $EP(t,X;T)$ denotes the (approximated) value of a European basket-put with maturity $T$ at time $t$. Recall that the family $(\psi_{k})$ serves as the regression basis for the continuation value. Further we choose
$$\left\{1,\left(X_t^d\frac{\partial EP(t,X_t;T_{j+1})}{\partial X_t^d}\right)_{1 \leq d \leq D}, \left(X_t^d\frac{\partial EP(t,X_t;T_J)}{\partial X_t^d}\right)_{1 \leq d \leq D} \right\}$$
as a regression basis $(\varphi_{k})$ spanning the family of the surely optimal martingales.
Since there is no closed-form formula for the still-alive European basket-put, we use the moment-matching method to approximate their values (see e.g. \citet{Brigo}, and \citet{Lor}). To this end, Let $\D S_t = \frac{X_t^1+\ldots+X_t^D}{D},$ and consider another asset $G_t$ whose risk-neutral dynamic follows
\[
dG_t = r G_t dt + \tilde{\sigma} G_t d W_t^1,
\]
where  $\tilde{\sigma}$ is a constant. The value of the European put on this asset  can be easily computed by the well-known Black-Scholes formula, that is,
\begin{equation} \label{num_bs}
E[ e^{-rT}(K-G_T)^{+}] = BS(G_0,r,\tilde{\sigma},K,T).
\end{equation}
If $S_T$ and $G_T$ have the same moments up to two, then the Black-Scholes price in \eqref{num_bs} can be regarded as a good approximation for the value of the European basket-put $E\left(e^{-rT}(K-S_T)^{+}\right),$
for details see \cite{Lor}. Since
\[
\begin{split}
E(S_T) &= \frac{1}{D}\sum_{d=1}^DX_0^d e^{rT},\\
E(S_T^2) &=\frac{1}{D^2}e^{2rT} \left( \sum_{i,j=1}^D X_0^i X_0^j \exp(1_{i=j}\sigma^2 T) \right)
\end{split}
\]
and
\[
E(G_T) = G_0 e^{rT},\qquad
E(G_T^2) = G_0^2 e^{2rT+\tilde{\sigma}^2T},
\]
we can simply set
\[
G_0 = \frac{1}{D}\sum_{d=1}^DX_0^d
\]
and
\[
\tilde{\sigma}^2 = \frac{1}{T} \ln \left(\frac{1}{(\sum_{d=1}^DX_0^d)^2} \sum_{i,j=1}^D X_0^i X_0^j \exp(1_{i=j}\sigma^2 T) \right).
\]
The European deltas can be approximated by
\[
\frac{\partial BS}{\partial G_0}\frac{\partial G_0 }{\partial X_0^d} = -\mathcal{N}(-d_1) \frac{1}{D}, \quad d = 1,\ldots,D,
\]
where $\D d_1= \frac{\ln(\frac{G_0}{K})+(r+\frac{\tilde{\sigma}^2}{2})T}{\tilde{\sigma}\sqrt{T}}$ and $\mathcal{N}$ denotes the cumulative standard normal distribution function. These formulas are straightforwardly extended to the pricing at times $t>0$.

\medskip

The numerical results are shown in Table~\ref{Tab1}. We use 1000 paths for estimating the surely optimal martingale and the continuation function via the regression procedure. Another 300000 paths are used to compute the lower bound and 100000 paths are used to compute the upper bound. Note that we have chosen a relatively small number of samples (1000) for estimating the martingale in the regression procedure. We do so because on the assumption that the family of uniformly integrable martingales is rich enough, the arguments leading to \eqref{SaS} yield that only a small number of samples are required for identifying a good approximation to a surely optimal martingale. We compare our results to the price intervals obtained in \citet{BKS2} which are displayed in the last column of Table \ref{Tab1}. In our C++ implementation, the run-times for computing one set of lower and upper bounds are in the range of 15-20 minutes.


\medskip


\begin{table}[ht]
\caption{Lower and upper bounds for Bermudan basket-put on 5 assets with parameters $r = 0.05$, $\delta=0$, $\sigma =0.2$, $K=100$, $T=3$ and different $J$ and $x_0$ }
\label{Tab1}
\medskip%
\centering
\begin{tabular}
[c]{|c|c|c|c|c|}\hline
$J$ & $x_{0}$ & Low (SE)	& Up (SE) & BKS Price Interval\\
\hline
& 90 & 10.000 (0.000) & 10.000 (0.000) &  [10.000, 10.004] \\
3 & 100 & 2.164 (0.007) & 2.172 (0.001)  & [2.154, 2.164] \\
& 110 & 0.539 (0.004) & 0.551 (0.001) &   [0.535, 0.540] \\
\hline
& 90 & 10.000 (0.000) & 10.000 (0.000) &   [10.000, 10.000] \\
6 & 100 & 2.407 (0.006) & 2.432 (0.001) & [2.359, 2.412] \\
& 110 & 0.573 (0.003) & 0.609 (0.001) &  [0.569, 0.580] \\
\hline
& 90 & 10.000 (0.0000) & 10.008 (0.0003) &  [10.000, 10.005] \\
9 & 100 & 2.475 (0.0063)  & 2.522 (0.0013)  &  [2.385, 2.502]  \\
& 110 & 0.5915 (0.0034) & 0.6353 (0.0009)&    [0.577, 0.600] \\
\hline
\end{tabular}
\end{table}

\subsection{Bermudan max-call}
\label{secbmc}
We use the same parameter values as in Section~\ref{secbsp} except $\delta=0.1$ and $D=2$ or $5$. As in the previous example we use European (call) options in the basis
 $(\psi_k)$ and the corresponding deltas in the basis
$(\varphi_k).$
The value of the European max-call option is computed by the following formula (\citet{John}),
\begin{align}\label{eq:max_call_price}
C_{max} &= \sum_{l=1}^{D}X_{0}^{l}\frac{e^{-\delta T}}{\sqrt{2\pi}%
}\int_{(-\infty,d_{+}^{l}]}\exp[-\frac{1}{2}z^{2}]%
{\displaystyle\prod_{\substack{l^{\prime}=1\\l^{\prime}\neq l}}^{D}}
\mathcal{N}\left(
\frac{\ln\frac{X_{0}^{l}}{X_{0}^{l^{\prime}}}}{\sigma
\sqrt{T}}-z+\sigma\sqrt{T}\right)  dz \nonumber\\
&  -K e^{-rT}+K e^{-rT}%
{\displaystyle\prod_{l=1}^{D}}
\left(  1-\mathcal{N}\left(  d_{-}^{l}\right)  \right),
\end{align}
where
\[
d_{-}^{l}:=\frac{\ln\frac{X_{0}^{l}}{K}+(r-\delta-\frac{\sigma^{2}}{2}%
)T}{\sigma\sqrt{T}},\quad d_{+}^{l}=d_{-}^{l}+\sigma\sqrt{T}.%
\]
Moreover, straightforward computations reveal that the deltas are given by
\begin{align}
\label{eq:max_call_deltas}
\frac{\partial C_{max}}{\partial X^l_0} &=  \frac{e^{-\delta T}}{\sqrt{2\pi}}\int_{(-\infty,d_{+}^{l}]}\exp[-\frac{1}{2}z^{2}] {\displaystyle\prod_{\substack{l^{\prime}=1\\l^{\prime}\neq l}}^{D}}\mathcal{N}\left(\frac{\ln\frac{X_{0}^{l}}{X_{0}^{l^{\prime}}}}{\sigma\sqrt{T}}-z+\sigma\sqrt{T}\right)  dz,
\end{align}
and that $C_{max}$ satisfies the linear homogeneity\footnote{Compare also with \cite[eq. (9)]{John}.}
\begin{align}\label{eq:max_call_linear}
C_{max} &= \sum_{l=1}^D X_0^l \frac{\partial C_{max}}{\partial X^l_0} + K \frac{\partial C_{max}}{\partial K}.
\end{align}


\begin{table}[ht]
\caption{Lower and upper bounds for Bermudan max-call with parameters $r = 0.05$, $\delta=0.1$, $\sigma =0.2$, $K=100$, $T=3$ and different $D$ and $x_0$.  }
\label{Tab2a}
\medskip%
\centering
\begin{tabular}
[c]{|c|c|c|c|c|}\hline
$D$& $x_{0}$ & Low (SE) & Up (SE) & A\&B price interval \\
\hline
   & 90 &  8.0556 (0.0219) & 8.15655 (0.0034) & [8.053, 8.082] \\
 2 & 100 & 13.8850  (0.0276) &  14.0293 (0.0044) & [13.892, 13.934] \\
   & 110 & 21.3671  (0.0319) &  21.5319 (0.0048) & [21.316, 21.359] \\
\hline
   & 90  &   16.5973 (0.0296) & 16.7963 (0.0058) & [16.602, 16.655]\\
 5 & 100 &   26.1325 (0.0356) & 26.3803 (0.0072) & [26.109, 26.292]\\
   & 110 &   36.7348 (0.0403) &  37.0856 (0.0082) & [36.704, 36.832]\\
\hline
\end{tabular}
\end{table}

The numerical results are shown in Table~\ref{Tab2a}. They are
based on 1000 paths for the regression procedure, 300000 paths for computing the lower bound  and 100000 paths for computing the upper bound. As before, we have chosen a relatively small number of samples (1000) for estimating the martingale in the regression procedure. This is again allowed because the arguments leading to \eqref{SaS} and the assumption that the choice of the basis functions indeed equips us with a rich enough family of uniformly integrable martingales yield that only a small number of samples are required for identifying a good approximation to a surely optimal martingale. The integral expressions from \eqref{eq:max_call_price} and \eqref{eq:max_call_deltas} are numerically evaluated using a simple adaptive Gauss-Kronrod procedure with $31$ points. The price intervals in the last column are quoted from \citet{AB}. In our C++ implementation, for each set of lower and upper bounds, we observe run-times that are in the range of 10-25 minutes, with the longer computation times for the 5-dimensional case.


\subsection*{Concluding remark}
The numerical results presented in Tables~\ref{Tab1}, \ref{Tab2a} due to our new algorithm
may be considered as very satisfactory given the decreased computation times (which are in the order of minutes in a C++ compiled implementation). In this respect it should be noted that computing upper bounds (in a rather generic way) in order of minutes is a considerable improvement compared to  \citet{BKS2}, whose upper bounds are computed with nested Monte Carlo simulation requiring higher computation time, and of comparable range to \citet{BBS}. Moreover, the algorithm delivers fast and surprisingly good lower bounds while the upper bounds are about the same range as the ones obtained with the algorithm in \citet{BBS}. Needless to say that, as
for the method of \citet{BBS}, the performance of the here presented algorithm will highly depend on the choice of the basis functions. An in depth treatment of this issue is considered beyond scope however.

\section{Appendix}

We present in this section some well-known facts from theory of empirical processes which are used to establish the relation \eqref{esti} in Section \ref{secalg0}.

\medskip

Let $\{\vartheta^{q}:q\in Q\}$ be a family of random variables and let for
each $q\in Q,$ $\vartheta^{q,1},...,\vartheta^{q,N}$ be i.i.d. samples of
$\vartheta^{q}.$ For each $q\in Q$ we consider the unbiased variance estimator%
\[
\Var^{(N)}\,\vartheta^{q}:=\frac{1}{N-1}\sum_{n=1}^{N}\left(  \vartheta
^{q,n}-\overline{\vartheta_{N}^{q}}\right)  ^{2}\text{ \ \ with \ \ }%
\overline{\vartheta_{N}^{q}}:=\frac{1}{N}\sum_{n=1}^{N}\vartheta^{q,n},
\]
hence $E\,\Var^{(N)}\,\vartheta^{q}=\Var\,\vartheta^{q}$ for $q\in Q.$ From
standard statistical theory it is well known that with $\mu^{q}:=E\vartheta
^{q}$
\[
\Var\left(  \Var^{(N)}\,\vartheta^{q}\right)  \leq\frac{1}{N}E\left(
\vartheta^{q}-\mu^{q}\right)  ^{4}=\frac{\left(  \Var\,\vartheta^{q}\right)
^{2}}{N}E\left(  \frac{\vartheta^{q}-\mu^{q}}{\sqrt{\Var\,\vartheta^{q}}%
}\right)  ^{4}.
\]
Now, as a mild condition we assume that,%
\[
E\left(  \frac{\vartheta^{q}-\mu^{q}}{\sqrt{\Var\,\vartheta^{q}}}\right)
^{4}\leq C\text{ \ \ for all }q\in Q.
\]
For example, this holds if $E\exp\left[  \lambda\left\vert
\vartheta^{q}\right\vert \right]  <\infty$ for some $\lambda>0$ and all $q\in
Q.$ We so have in particular%
\[
\Var\left(  \Var^{(N)}\,\vartheta^{q_{N}^{\circ}}\right)  \leq\frac{C}{N}\left(
\Var\,\vartheta^{q_{N}^{\circ}}\right)  ^{2}\text{ \ \ and \ \ }\Var\left(
\Var^{(N)}\,\vartheta^{q^{\circ}}\right)  \leq\frac{C}{N}\left(  \Var\,\vartheta
^{q^{\circ}}\right)  ^{2},
\]
while, strictly speaking,  the randomness of $\vartheta^{q_{N}^{\circ}}$ in the first inequality is ignored. However,  by considering a next from $\vartheta^{q_{N}^{\circ}}$ independent sample, we can show that this is not really essential (the details would go beyond the scope of the optional
analysis of Section~\ref{secalg0} and are therefore omitted).
From standard empirical probability theory it now follows that for any (small) $0\leq\alpha\ll1,$ there is a suitable quantile  coefficient
$c_{\alpha}$ (particularly not depending on $N$) such that
\begin{align*}
&\mathbb{P}\Bigg[  \Var^{(N)}\,\vartheta^{q}\leq \Var\,\vartheta^{q}\left(
1+c_{\alpha}\sqrt{\frac{C}{N}}\right),\Var\,\vartheta^{q}\leq \Var^{(N)}\,\vartheta^{q}+c_{\alpha}\sqrt{\frac{C}{N}}\Var\,\vartheta^{q}  \Bigg] \\
&\geq 1-\alpha,
\end{align*}
for $q\in\{q_{N}^{\circ},q^{\circ}\}$ which implies (\ref{esti}).

\subsection*{Acknowledgements}

J.S. is grateful to Denis Belomestny for interesting discussions on this
topic, and to Christian Bender for an inspiring remark, which led to
counterexample Example~\ref{counter1}.

\end{document}